\documentclass{article}

\usepackage{fullpage}
\usepackage{amsthm,amssymb,amsmath}  
\usepackage{xspace,enumerate}
\usepackage[capitalise]{cleveref}
\usepackage[utf8]{inputenc}
\usepackage{thmtools}
\usepackage{thm-restate}
\usepackage{authblk}

  \theoremstyle{plain}
  \newtheorem{theorem}{Theorem}
  \newtheorem{lemma}{Lemma}  
  \newtheorem{corollary}[theorem]{Corollary}  
  \newtheorem{fact}{Fact}
  
  \theoremstyle{definition}

\title{Optimal Computation of Overabundant Words}
\author[1]{Yannis Almirantis}
\author[2]{Panagiotis Charalampopoulos}
\author[2]{Jia Gao}
\author[2]{Costas~S.~Iliopoulos}
\author[2]{Manal Mohamed}
\author[2]{Solon~P.~Pissis}
\author[3]{Dimitris Polychronopoulos}

\affil[1]{National Center for Scientific Research Demokritos, Athens, Greece\\
\texttt{yalmir@bio.demokritos.gr}}
\affil[2]{Department of Informatics, King's College London, UK\\
\texttt{[panagiotis.charalampopoulos,jia.gao,costas.iliopoulos\\manal.mohamed,solon.pissis]@kcl.ac.uk}}
\affil[3]{Computational Regulatory Genomics, Institute of Clinical Sciences (ICS), Faculty of Medicine, Imperial College London, Du Cane Road, London W12 0NN\\
\texttt{dpolychr@imperial.ac.uk}}

\date{\vspace{-5ex}}

\usepackage[table]{xcolor}
\usepackage{arydshln}
\usepackage[noend]{algpseudocode}
 \usepackage{algorithmicx}
    \usepackage[ruled]{algorithm}
    \usepackage{algpseudocode}

\algnewcommand{\LineComment}[1]{\State \(\triangleright\) #1}

\usepackage{microtype}
\usepackage{tikz}
\usetikzlibrary{trees}
\usetikzlibrary{decorations.pathmorphing}
\usetikzlibrary{decorations.markings}
\usetikzlibrary{decorations.pathmorphing,shapes}
\usetikzlibrary{calc,decorations.pathmorphing,shapes}
  \usepackage{comment}
  \usepackage{pgfplots}
  \usepackage{forest}
  \usepackage{xspace}
  \usepackage{enumerate}
  \usepackage{multirow}
  \usepackage{todonotes}
  \usepackage{thmtools}
  \usepackage{thm-restate}
  \usepackage[capitalise]{cleveref}
  \usepackage{graphics,adjustbox}
  \usepackage{epsfig}
  \usepackage{graphicx}
  \usepackage{epstopdf}
  \usepackage{url}
  \usepackage{verbatim}
  \usetikzlibrary{trees, arrows, shapes}
\tikzset{
  treenode/.style = {align=center, inner sep=2pt, text centered,
    font=\sffamily},
  arn_r/.style = {treenode, circle, black, font=\sffamily\bfseries, draw=black,
    text width=1.5em},
    arn_t/.style = {treenode, circle, black, thick, double, font=\sffamily\bfseries, draw=black,
    text width=1.5em},
  every edge/.append style={anchor=south,auto=falseanchor=south,auto=false,font=3.5 em},
}
\newcommand{\std}{\textsl{dev}}

\crefname{lem}{Lemma}{Lemmas}

\def\dd{\mathinner{.\,.}}
\newcommand{\cO}{\mathcal{O}}

 \newcommand{\defproblem}[3]{
  \vspace{2mm}
\noindent\fbox{
  \begin{minipage}{0.96\textwidth}
  #1\\
  {\bf{Input:}} #2  \\
  {\bf{Output:}} #3
  \end{minipage}
  }
  \vspace{2mm}
}

\begin{document}
\maketitle

\begin{abstract}
The observed frequency of the longest proper prefix, the longest proper suffix, and the longest infix of a word $w$ in a given sequence $x$ can be used for classifying $w$ as \textit{avoided} or \textit{overabundant}. The definitions used for the expectation and deviation of $w$ in this statistical model were described and biologically justified by Brendel et al. (J Biomol Struct Dyn 1986). We have very recently introduced a time-optimal algorithm for computing all avoided words of a given sequence over an integer alphabet (Algorithms Mol Biol 2017). In this article, we extend this study by presenting an $\mathcal{O}(n)$-time and $\mathcal{O}(n)$-space algorithm for computing all overabundant words in a sequence $x$ of length $n$ over an integer alphabet. Our main result is based on a new {\em non-trivial} combinatorial property of the suffix tree $\mathcal{T}$ of $x$: the number of distinct factors of $x$ whose longest infix is the label of an explicit node of $\mathcal{T}$ is no more than $3n-4$. We further show that the presented algorithm is time-optimal by proving that $\mathcal{O}(n)$ is a tight upper bound for the number of overabundant words. Finally, we present experimental results, using both synthetic and real data, which justify the {\em effectiveness} and {\em efficiency} of our approach in practical terms.
\end{abstract}

\section{Introduction}\label{sec:intro}

Brendel et al.~in \cite{brendel1986linguistics} initiated research into the linguistics of nucleotide sequences that focused on  the concept of words in continuous languages---languages devoid of blanks---and introduced an operational definition of words. The authors suggested a method to measure, for each possible word $w$ of length $k$, the deviation of its observed frequency $f(w)$ from the expected frequency $E(w)$ in a given sequence $x$. The observed frequency of the longest proper prefix, the longest proper suffix, and the longest infix of $w$ in $x$ were used to measure $E(w)$. The values of the deviation, denoted by $\std(w)$, were then used to identify words that are \textit{avoided} or \textit{overabundant} among all possible words of length $k$. The typical length of avoided (or of overabundant) words of the nucleotide language was found to range from 3 to 5 (tri- to pentamers). The statistical significance of the avoided words was shown to reflect their biological importance. That work, however, was based on the very limited sequence data available at the time: only DNA sequences from two viral and one bacterial genomes were considered. Also note that the range of typical word length $k$ might change when considering eukaryotic genomes, the complex dynamics and function of which are expected to impose more demanding roles to avoided or overabundant words of nucleotides. 

To this end, in~\cite{AMOB2017}, we presented an $\cO(n)$-time and $\cO(n)$-space algorithm for computing all avoided words of length $k$ in a sequence of length $n$ over a fixed-sized alphabet. For words over an integer alphabet of size $\sigma$, the algorithm requires time $\cO(\sigma n)$, which is optimal for sufficiently large $\sigma$. We also presented a time-optimal $\cO(\sigma n)$-time algorithm to compute all avoided words (of any length) in a sequence of length $n$ over an integer alphabet of size $\sigma$. We provided a tight asymptotic upper bound for the number of avoided words over an integer alphabet and the expected length of the longest one. We also proved that the same asymptotic upper bound is tight for the number of avoided words of fixed length $k$ when the alphabet is sufficiently large. The authors in~\cite{lgdrs, mslq, Apostolico:2004:VDU:987206.987210} studied a similar notion of {\em unusual words}---based on different definitions than the ones Brendel et al. use for expectation and deviation---focusing on the factors of a sequence; based on Brendel et al.'s definitions, we focus on any word over the alphabet. More recently, space-efficient detection of unusual words has also been considered~\cite{DBLP:conf/spire/BelazzouguiC15}; such avoidances is becoming an interesting line of research~\cite{Rusinov2015}. 

In this article, we wish to complement our study in~\cite{AMOB2017} by focusing on overabundant words. The motivation comes from molecular biology. Genome dynamics, i.e.~the molecular mechanisms generating random mutations in the evolving genome, are quite complex, often presenting self-enhancing features. Thus, it is expected to often give rise to words of nucleotides which will be overabundant, i.e.~being present at higher amounts than expected on the basis of their longest proper prefix, longest proper suffix, and longest infix frequencies. One specific such mechanism, which might generate overabundant words, is the following: it is well-known that in a genomic sequence of an initially random composition, the existing relatively long homonucleotide tracts present a higher frequency of further elongation than the frequency expected on the basis of single nucleotide mutations~\cite{Levinson01051987}; that is, they present a sort of autocatalytic self-elongation. This feature, in combination with the much higher frequency of transition {\em vs.} transversion mutation events, generates overabundant words which are homopurinic or homopurimidinic tracts. It is also anticipated that the overabundance of homonucleotide tracts will strongly differentiate between conserved and non-conserved parts of the genome. While this phenomenon is largely free to act within the non-conserved genomic regions, and thus it is expected to generate there large amounts of overabundant words, it is hindered in the conserved genomic regions due to selective constraints.

{\bf Our Contributions.} Analogously to avoided words~\cite{brendel1986linguistics,avoidedNAR,AMOB2017}, many different models and algorithms exist for identifying words that are in abundance in a given sequence; see for instance~\cite{Karlin,Denise:2001:ASS:645906.673108}. In this article, we make use of the biologically justified model introduced by Brendel et al.~\cite{brendel1986linguistics} and, by proving non-trivial combinatorial properties, we show that it admits {\em efficient} computation for overabundant words as well. We also present experimental results, using both synthetic and real data, which further highlight the {\em effectiveness} of this model. The computational problem can be described as follows. Given a sequence $x$ of length $n$ and a real number $\rho > 0$, compute the set of $\rho$-overabundant words, i.e.~all words $w$ for which $\std(w) \geq \rho$. We present an $\cO(n)$-time and $\cO(n)$-space algorithm for computing all $\rho$-overabundant words (of any length) in a sequence $x$ of length $n$ over an integer alphabet. This result is based on a combinatorial property of the suffix tree $\mathcal{T}$ of $x$ that we prove here: the number of distinct factors of $x$ whose longest infix is the label of an explicit node of $\mathcal{T}$ is no more than $3n-4$. We further show that the presented algorithm is time-optimal by proving that $\mathcal{O}(n)$ is a tight upper bound for the number of $\rho$-overabundant words. Finally, we pose an open question of combinatorial nature on the maximum number $\textsf{OW}(n,\sigma)$ of overabundant words that a sequence of length $n$ over an alphabet of size $\sigma>1$ can contain.

\section{Terminology and Technical Background}\label{sec:prel} 
\subsection{Definitions and Notation}
We begin with basic definitions and notation, generally following~\cite{CHL07}. Let $x=x[0]x[1]\dd x[n-1]$ be a \textit{word} of \textit{length} $n=|x|$ over a finite ordered \textit{alphabet} $\Sigma$ of size $\sigma$, i.e.~$\sigma = |\Sigma|$. In particular, we consider the case of an \textit{integer alphabet}; in this case each letter is replaced by its rank such that the resulting word consists of integers in the range $\{1,\ldots,n\}$. In what follows we assume without loss of generality that $\Sigma=\{0,1,\ldots,\sigma-1\}$. We also define $\Sigma_x$ to be the alphabet of word $x$ and $\sigma_x=|\Sigma_x|$. For two positions $i$ and $j$ on $x$, we denote by $x[i\dd j]=x[i]\dd x[j]$ the \textit{factor} (sometimes called \textit{subword}) of $x$ that starts at position $i$ and ends at position $j$ (it is empty if $j < i$), and by $\varepsilon$ the \textit{empty word}, word of length 0.  We recall that a prefix of $x$ is a factor that starts at position 0 ($x[0\dd j]$) and a suffix is a factor that ends at position $n-1$ ($x[i\dd n-1]$), and that a factor of $x$ is a \textit{proper} factor if 
it is not $x$ itself. A factor of $x$ that is neither a prefix nor a suffix of $x$ is called an $\textit{infix}$ of $x$. We denote the reverse word of $x$ by $\textsf{rev}(x)$, i.e.~$\textsf{rev}(x)=x[n-1]x[n-2]\dd x[1]x[0]$. We say that $x$ is \textit{a power} of a word $y$ if there exists a positive integer $k$, $k>1$, such that $x$ is expressed as $k$ consecutive concatenations of $y$; we denote that by $x=y^k$.

Let $w=w[0]w[1]\dd w[m-1]$ be a word, $0<m\leq n$. 
  We say that there exists an \textit{occurrence} of $w$ in $x$, or, more 
simply, that $w$ \textit{occurs in} $x$, if $w$ is a factor of $x$, which we denote by $w \preceq x$.
  Every occurrence of $w$ can be characterised by a starting position in $x$. 
  Thus we say that $w$ occurs at \textit{position} $i$ in $x$ when $w=x[i \dd i + m - 1]$. Further, let $f(w)$ denote the \textit{observed frequency}, that is, the number of occurrences of a non-empty word $w$ in word $x$. If $f(w) = 0$ for some word $w$, then $w$ is called \textit{absent} (which is denoted by $w \not\preceq x$), otherwise, $w$ is called \textit{occurring}.

By $f(w_p)$, $f(w_s)$, and $f(w_i)$ we denote the observed frequency of the longest proper prefix $w_p$, suffix $w_s$, and infix $w_i$ of $w$ in $x$, respectively. 
We can now define the \textit{expected frequency} of word $w$, $|w|>2$, in $x$ as in Brendel et al.~\cite{brendel1986linguistics}:  
\begin{equation} \label{eq:1}
E(w) =  \frac {f(w_p) \times f(w_s)}{f(w_i)}, \text{ if~ } f(w_i) >0;  \text{~else~}  E(w) = 0.
\end{equation}
The above definition can be explained intuitively as follows. Suppose we are given $f(w_p)$, $f(w_s)$, and $f(w_i)$. Given an occurrence of $w_i$ in $x$, the probability of it being preceded by $w[0]$ is $\frac {f(w_p)}{f(w_i)}$ as $w[0]$ precedes exactly $f(w_p)$ of the $f(w_i)$ occurrences of $w_i$. Similarly,  this occurrence of $w_i$ is also an occurrence of $w_s$ with probability $\frac {f(w_s)}{f(w_i)}$. Although these two events are not always independent, the product $\frac {f(w_p)}{f(w_i)} \times \frac {f(w_s)}{f(w_i)}$ gives a good approximation of the probability that an occurrence of $w_i$ at position $j$ implies an occurrence of $w$ at position $j-1$. It can be seen then that by multiplying this product by the number of occurrences of $w_i$ we get the above formula for the expected frequency of $w$.

\noindent Moreover, to measure the deviation of the observed frequency of a word $w$ from its expected frequency in $x$, we define the {\it deviation} ($\chi^2$ test) of $w$ as:
\begin{equation} \label{eq:2}
 \std(w) = \frac {f(w)-E(w)}{\max\{ \sqrt {E(w)}, 1\}}.
\end{equation}
\noindent For more details on the {\it biological} justification of these definitions see~\cite{brendel1986linguistics} and~\cite{AMOB2017}.

\pgfmathdeclarefunction{gauss}{2}{%
  \pgfmathparse{1/(#2*sqrt(2*pi))*exp(-((x-#1)^2)/(2*#2^2))}%
}

\pgfmathdeclarefunction{gauss2}{2}{%
  \pgfmathparse{1/(#2*sqrt(2*pi))*exp(-((x-#1)^2)/(2*#2^2))*3/2}%
}

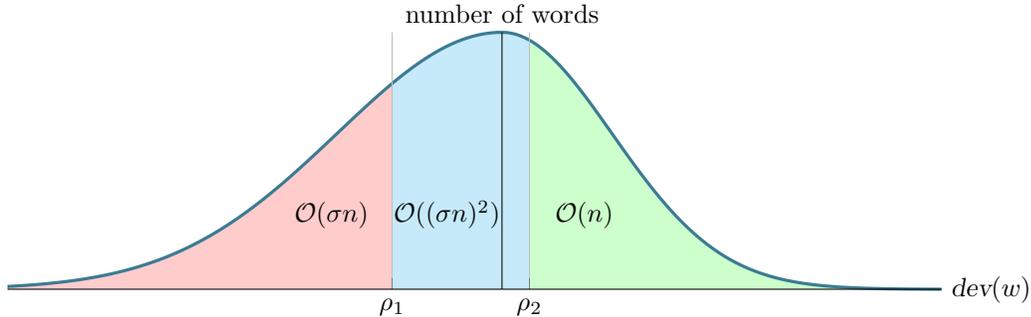
\begin{figure}
\begin{center}
\begin{tikzpicture}
\begin{axis}[
  no markers, domain=-9:8, samples=100,
  axis lines*=center, xlabel=$\std(w)$, ylabel=number of words,
  every axis y label/.style={at=(current axis.above origin),anchor=south},
  every axis x label/.style={at=(current axis.right of origin),anchor=west},
  height=5cm, width=14cm,
  xtick={-2,0.5},
  xticklabels={$\rho_1$,$\rho_2$},
  ytick=\empty,
  enlargelimits=false, clip=false, axis on top,
  grid = major
  ]
  \addplot [fill=red!20, draw=none, domain=-9:-2] {gauss2(0,3)} \closedcycle;
    \addplot [fill=cyan!20, draw=none, domain=-2:0] {gauss2(0,3)} \closedcycle;
      \addplot [fill=cyan!20, draw=none, domain=0:0.5] {gauss(0,2)} \closedcycle;
  \addplot [fill=green!20, draw=none, domain=0.5:8] {gauss(0,2)} \closedcycle;
  \addplot [very thick,cyan!50!black][domain=0:8] {gauss(0,2)};
    \addplot [very thick,cyan!50!black][domain=-9:0] {gauss2(0,3)};

    \node[above] at (590,40) {$\cO(\sigma n)$};
    \node[above] at (800,40) {$\cO((\sigma n)^2)$};
	\node[above] at (1050,40) {$\cO(n)$};

\end{axis}
\end{tikzpicture}
\end{center}
\caption{For a word $x$, the words for which $\std(w)$ is defined are the ones of the form $w=aub$, where $u$ is a factor of $x$ and $a,b \in \Sigma$, not necessarily distinct. There are $\cO(n^2)$ distinct factors in a word of length $n$ and for each of these we obtain $\sigma ^2$ words of this form. We have shown that the $\rho_1$-avoided words are $\cO(\sigma n)$~\cite{AMOB2017}. In this article, we show that the $\rho_2$-overabundant ones are $\cO(n)$.}
\label{fig:bell}
\end{figure}

Using the above definitions and two given thresholds, we can classify a word $w$ as either \textit{avoided}, \textit{common}, or \textit{overabundant} in $x$. In particular, for  two given thresholds $\rho_1 < 0$ and $\rho_2 > 0$, a word $w$ is called $\rho_1$-\textit{avoided} if $\std(w) \leq \rho_1$, $\rho_2$-\textit{overabundant} if $\std(w) \geq \rho_2$, and $(\rho_1,\rho_2)$-\textit{common} otherwise (see Figure~\ref{fig:bell}). We have very recently shown that the number of $\rho_1$-avoided words is $\mathcal{O}(\sigma n)$, and have introduced a time-optimal algorithm for computing all of them in a given sequence over an integer alphabet~\cite{AMOB2017}. In this article, we show that the number of $\rho_2$-overabundant words is $\cO(n)$, and study the following computational problem. 

{\defproblem{\textsc{AllOverabundantWordsComputation}}{A word $x$ of length $n$ and a real number $\rho > 0$}{All $\rho$-overabundant words in $x$}}


\subsection{Suffix Trees}
In our algorithms, suffix trees are used  extensively as computational tools. For a general introduction to suffix trees see~\cite{CHL07}.

The \textit{suffix tree} $\mathcal{T}(x)$ of a non-empty word $x$ of length $n$ is a compact trie representing all suffixes of $x$. The nodes of the trie which become nodes of the suffix
tree are called {\it explicit} nodes, while the other nodes are called {\it implicit}. Each edge
of the suffix tree can be viewed as an upward maximal path of implicit nodes starting with an explicit node. Moreover, each node belongs to a unique path of that kind. Then, each node of the trie can be represented in the suffix tree by the edge it belongs to and an index within the corresponding path.

We use $\mathcal{L}(v)$ to denote the \textit{path-label} of a node $v$, i.e., the concatenation of the edge labels along the path from the root to $v$. We say that $v$ is  path-labelled $\mathcal{L}(v)$. Additionally, $\mathcal{D}(v)= |\mathcal{L}(v)|$ is used to denote the \textit{word-depth} of node $v$. Node $v$ is a \textit{terminal} node  if and only if $\mathcal{L}(v) = x[i \dd n-1]$, $0 \leq i < n$; here $v$ is also labelled with index $i$. It should be clear that each occurring word $w$ in $x$ is uniquely represented by either an explicit or an implicit node of $\mathcal{T}(x)$. The \textit{suffix-link} of a node $v$ with path-label $\mathcal{L}(v)= \alpha y$ is a pointer to the node path-labelled $y$, where $\alpha \in \Sigma$ is a single letter and $y$ is a word. The suffix-link of $v$ exists if $v$ is a non-root internal node of $\mathcal{T}(x)$. 

In any standard implementation of the suffix tree, we assume that each node of the suffix tree is able to access its parent. Note that once $\mathcal{T}(x)$ is constructed, it can be traversed in a depth-first manner to compute the word-depth $\mathcal{D}(v)$ for each node $v$. Let $u$ be the parent of $v$. Then the word-depth $\mathcal{D}(v)$ is computed by adding  $\mathcal{D}(u)$ to the length of the label of edge $(u,v)$. If $v$ is the root then $\mathcal{D}(v) = 0$. Additionally, a depth-first traversal of $\mathcal{T}(x)$ allows us to count,  for  each  node $v$, the number of terminal nodes in the subtree rooted at $v$, denoted by $\mathcal{C}(v)$, as follows. When internal node $v$  is visited, $\mathcal{C}(v)$ is computed by adding up $\mathcal{C}(u)$ of all the nodes $u$, such that $u$ is a child of $v$, and then $\mathcal{C}(v)$ is incremented by 1 if $v$ itself is a terminal node. If a node $v$ is a leaf  then $\mathcal{C}(v) = 1$.

We assume that the terminal nodes of $\mathcal{T}(x)$ have suffix-links as well. 
We can either store them while building $\mathcal{T}(x)$ 
or just traverse it once and construct an array $\textit{node}[0 \dd n-1]$ such that $\textit{node}[i]=v$ if $\mathcal{L}(v)=x[i \dd n-1]$. 
We further denote by $\textsc{Parent}(v)$ the parent of a node $v$ in $\mathcal{T}(x)$ and by $\textsc{Child}(v,\alpha)$ the explicit node that is obtained from $v$ by traversing the outgoing edge whose label starts with $\alpha \in \Sigma$. A batch of $q$ $\textsc{Child}(v,\alpha)$ queries can be answered off-line in time $\cO(n+q)$ for a word $x$ over an integer alphabet (via radix sort).

\section{Combinatorial Properties}\label{Section-Properties}

In this section, we prove some properties that are useful for designing the time-optimal algorithm presented in the next section. 

\begin{fact}\label{fac:common}
Given a word $x$ of length $n$ over an alphabet of size $\sigma$, the number of words $w$ for which $\std(w)$ is defined is $\mathcal{O}((\sigma n)^2)$.
\end{fact}
\begin{proof}
For a word $w$ over $\Sigma$, $\std(w)$ is only defined if $w_i \preceq x$. Hence the words $w$ for which $\std(w)$ is defined are of the form $aub$ for some non-empty $u \preceq x$ and $a,b \in \Sigma$. For each distinct factor $u \neq \varepsilon$ of $x$ there are $\sigma ^2$ words of the form $aub$, $a,b \in \Sigma$. Since there are $\cO(n^2)$ distinct factors in a word of length $n$, the fact follows.
\end{proof}

\begin{fact}\label{abs}
Every word $w$ that does not occur in $x$ and for which $\std(w)$ is defined has $\std(w) \leq 0$.
\end{fact}
\begin{proof}
For such a word we have that $E(w) \geq 0$ and that $f(w)=0$ and hence
$\std(w)=\frac{f(w)-E(w)}{\max\{{\sqrt{E(w)}, 1}\}} \leq 0$.
\end{proof}

\noindent {\it Na\"{\i}ve algorithm}. By using Fact~\ref{abs}, we can compute $\std(w)$, for each factor $w$ of $x$, thus solving Problem \textsc{AllOverabundantWordsComputation}. There are $\mathcal{O}(n^2)$
such factors, however, which make this computation inefficient.

\begin{fact}\label{fact:impl}
Given a factor $w$ of a word $x$, if $w_i$ corresponds to an implicit node in the suffix tree $\mathcal{T}(x)$, then so does  $w_p$.
\end{fact}

\begin{proof}
A factor $w'$ of $x$ corresponds to an implicit node $\mathcal{T}(x)$ if and only if every occurrence of it in $x$ is followed by the same unique letter $b \in \Sigma$.
Hence, since $w_p=aw_i$ for some $a \in \Sigma$, if $w_i$ is always followed by, say, $b \in \Sigma$, every occurrence of $w_p$ in $x$ must also always be followed 
by $b$. Thus $w_p$ corresponds to an implicit node as well.
\end{proof}

\begin{lemma}\label{expl}
If $w$ is a factor of a word $x$ and $w_i$ corresponds to an implicit node in $\mathcal{T}(x)$, then $\std(w)=0$.
\end{lemma}

\begin{proof}
If a word $w' \preceq x$ corresponds to an implicit node along the edge $(u, v)$ in $\mathcal{T}(x)$ and $\mathcal{L} (v)=w$ then the number of 
occurrences of $w'$ in $x$ is equal to that of $w$. 

If $w_i$ corresponds to an implicit node on edge $(u, v)$ it follows immediately that $f(w_i)=f(w_s)$, as either $w_s$ also corresponds to an implicit node in the same edge or $w_s=\mathcal{L} (v)$. 
In addition, from Fact~\ref{fact:impl} we have that $w_p$ is an implicit node as well and it similarly follows that $f(w_p)=f(w)$. We thus have $E(w)=\frac{f(w_p) \times f(w_s)}{f(w_i)}=f(w)$
and hence $\std(w)=\frac{f(w)-E(w)}{\max\{{\sqrt{E(w)}, 1}\}}=0$.
\end{proof}

Based on these properties, the aim of the algorithm in the next section is to find the factors of $x$ whose longest infix corresponds to an explicit node and check if they are $\rho$-overabundant. More specifically, for each explicit node $v$ in $\mathcal{T}(x)$, such that $\mathcal{L} (v)=y$, we aim at identifying the factors of $x$ that have $y$ as their longest infix (i.e.~factors of the form $ayb$, $a,b \in \Sigma$). We will do that by identifying the factors of $x$ that have $y$ as their longest proper suffix (i.e.~factors of the form $ay$, $a \in \Sigma$) and then checking for each of these the different letters that succeed it in $x$. Then we can check in time $\mathcal{O}(1)$ if each of these words is $\rho$-overabundant.

Note that the algorithm presented in Section~\ref{Section-Algorithm} is fundamentally different and in a sense more involved than the one presented in~\cite{AMOB2017} for the computation of {\em occurring} $\rho$-avoided words (note that a $\rho$-avoided word can be {\em absent}). This is due to the fact that for occurring $\rho$-avoided words we have the stronger property that $w_p$ must correspond to an explicit node.

\begin{theorem}\label{revbound}
Given a word $x$ of length $n$, the number of distinct factors of $x$ of the form $a y b$, where $a, b \in \Sigma$ and $y \neq \varepsilon$ is the label of an explicit node of $\mathcal{T}(x)$, is no more than $3n-2-2\sigma_x$.
\end{theorem}
\begin{proof}
Let $S$ be the set of all explicit or implicit nodes in $\mathcal{T}(x)$ of the form $yb$ such that $y$ is represented by an explicit node other than the root. We have at most $2n-2-\sigma_x$ of them; there are at most $2n-2$ edges in $\mathcal{T}(x)$, but $\sigma_x$ of them are outgoing from the root. For such a word $yb$, the number of factors of $x$ of the form $ayb$ is equal to the degree of the node representing $\textsf{rev}(yb)$ in $\mathcal{T}(\textsf{rev}(x))$. 

For every node in $S$, we obtain a distinct node in $\mathcal{T}(\textsf{rev}(x))$. Let us suppose that $k_1$ of these nodes are non-root internal explicit nodes, $k_2$ are leaves, and the rest $2n-2-\sigma_x-k_1-k_2$ are implicit nodes. Each internal explicit node $u$ contributes at most $deg(u)$, each leaf contributes $0$, and each implicit node contributes at most $1$. 

Hence the number of such words would be maximised if we obtained all the non-root internal explicit nodes and no leaves.
Let $\mathcal{T}(\textsf{rev}(x))$ have $m$ non-root internal explicit nodes. The resulting upper bound then is $\sum_{u \in \mathcal{T}(\textsf{rev}(x))\setminus \{root\}}^{} {deg(u)}+(2n-2-\sigma_x-m) \leq n+m-\sigma_x+(2n-2-\sigma_x-m)=3n-2-2\sigma_x$.

Note that $\sum_{u \in \mathcal{T}(\textsf{rev}(x))\setminus \{root\}}^{} {deg(u)}\leq n+m-\sigma_x$ since there are at most $n$ edges from explicit internal nodes to leaves and $m$ edges to other internal nodes; $\sigma_x$ of these are outgoing from the root.
\end{proof}

\begin{corollary}\label{lem:3n}
The $\rho$-overabundant words in a word $x$ of length $n$ are at most $3n-2-2\sigma_x$.
\end{corollary}
\begin{proof}
By Fact~\ref{abs}, Lemma~\ref{expl}, and symmetry, it follows that the $\rho$-overabundant words in $x$ are factors of $x$ of the form $a y b$, where $a, b \in \Sigma$, such that $y \neq \varepsilon$ is represented by an explicit node in $\mathcal{T}(x)$ and $\textsf{rev}(y)$ represented by an explicit node in $\mathcal{T}(\textsf{rev}(x))$. Hence they are a subset of the set of words considered in Theorem~\ref{revbound}.
\end{proof}

\begin{lemma} 
The $\rho$-overabundant words in a word $x$ of length $n$ over a binary alphabet (e.g. $\Sigma=\{\texttt{a}, \texttt{b}\}$) are no more than $2n-4$.
\end{lemma}

\begin{proof}
For every internal explicit node $u$ of $\mathcal{T}(x)$, other than the root, let $deg'(u)$ be $deg(u)+1$ if node $u$ is terminal and $deg(u)$ otherwise. The sum of $deg'(u)$ over the internal explicit non-root nodes of $\mathcal{T}(x)$ is no more than $2n-4$ (ignoring the case when $x={\alpha}^n, \alpha \in \Sigma$).
We will show that, for each such node, the $\rho$-overabundant words with $w_i=\mathcal{L}(u)$ as their longest proper infix are at most $deg'(u)$.

\begin{itemize}
\item \emph{Case I: $deg'(u)=2$.}
\begin{itemize}
\item \emph{Subcase 1: $deg(u)=1$.}
Node $u$ is terminal and it has an edge with label $\alpha$. We can then have at most 2 $\rho$-overabundant words with $w_i$ as their longest proper infix: $\texttt{a}w_i \alpha$ and $\texttt{b}w_i \alpha$.
\item \emph{Subcase 2: $deg(u)=2$.}
Node $u$ is not terminal and it has an edge with label $\texttt{a}$ and an edge with label $\texttt{b}$.
If only one of $\texttt{a}w_i$ and $\texttt{b}w_i$ occurs in $x$ we are done. If both of them occur in $x$ we argue as follows (irrespective of whether $w_i$ is also a prefix of $x$):

If $\texttt{a}w_i \texttt{a}$ is $\rho$-overabundant, then

$f(\texttt{a}w_i\texttt{a})-f(\texttt{a}w_i) \times f(w_i \texttt{a})/f(w_i) \geq \rho >0 \Rightarrow f(\texttt{a}w_i\texttt{a})/f(\texttt{a}w_i) > f(w_i \texttt{a})/f(w_i) \Leftrightarrow 1-f(\texttt{a}w_i\texttt{a})/f(\texttt{a}w_i) < 1-f(w_i \texttt{a})/f(w_i) \Leftrightarrow f(\texttt{a}w_i\texttt{b})/f(\texttt{a}w_i) < f(w_i \texttt{b})/f(w_i) \Leftrightarrow f(\texttt{a}w_i\texttt{b})-f(\texttt{a}w_i) \times f(w_i \texttt{b})/f(w_i) < 0$

and hence $\texttt{a}w_i\texttt{b}$ is not $\rho$-overabundant. (Similarly for $\texttt{b}w_i \texttt{a}$ and $\texttt{b}w_i \texttt{b}$.)
\end{itemize}

\item \emph{Case II: $deg'(u)=3$.}
Node $u$ is terminal and it has an edge with label $\texttt{a}$ and an edge with label $\texttt{b}$. If only one of $\texttt{a}w_i$ and $\texttt{b}w_i$ occurs in $x$ or if both of them occur in $x$, but $w_i$ is not a prefix of $x$, we can have at most 2 $\rho$-overabundant words with $w_i$ as the proper longest infix; this can be seen by looking at the node representing $\textsf{rev}(w_i)$ in $\mathcal{T}(\textsf{rev}(x))$, which falls in \emph{Case I}.

So we only have to consider the case where both $\texttt{a}w_i$ and $\texttt{b}w_i$ occur in $x$ and $w_i$ is a prefix of $x$. For this case, we assume without loss of generality that $\texttt{a} w_i$ is a suffix of $x$.
If $\texttt{a}w_i \texttt{a}$ is $\rho$-overabundant, then

$f(\texttt{a}w_i\texttt{a})-f(\texttt{a}w_i) \times f(w_i \texttt{a})/f(w_i) \geq \rho >0 \Rightarrow f(\texttt{a}w_i\texttt{a})/f(\texttt{a}w_i) > f(w_i \texttt{a})/f(w_i) \Leftrightarrow 1-f(\texttt{a}w_i\texttt{a})/f(\texttt{a}w_i) < 1-f(w_i \texttt{a})/f(w_i) \Leftrightarrow (f(\texttt{a}w_i\texttt{b})+1)/f(\texttt{a}w_i) < (f(w_i \texttt{b})+1)/f(w_i) \Rightarrow 
f(\texttt{a}w_i\texttt{b})/f(\texttt{a}w_i) < (f(w_i \texttt{b})/f(w_i) \Leftrightarrow
f(\texttt{a}w_i\texttt{b})-f(\texttt{a}w_i) \times f(w_i \texttt{b})/f(w_i) < 0$

and hence $\texttt{a}w_i\texttt{b}$ is not $\rho$-overabundant. Thus in this case we can have at most $3=deg'(u)$ $\rho$-overabundant words.
\end{itemize}

We can thus have at most $deg'(u)$ $\rho$-overabundant words for each internal explicit non-root node of $\mathcal{T}(x)$. This concludes the proof.
\end{proof}

\begin{lemma}\label{lem:opt}
The $\rho$-overabundant words in a word of length $n$ are $\mathcal{O}(n)$ and this bound is tight.
There exists a word over the binary alphabet with $2n-6$ $\rho$-overabundant words.
\end{lemma}
\begin{proof}
The asymptotic bound follows directly from Corollary~\ref{lem:3n}.
The tightness of the asymptotic bound can be seen by considering word $x=ba^{n-2}b$, $a,b \in \Sigma$, of length $n$ and some $\rho$ such that $0<\rho<1/n$.
Then for every prefix $w$ of $x$ of the form $ba^k$ or for every suffix $w$ of $x$ of the form $a^kb$, $2 \leq k \leq n-2$, we have that $f(w_p)=f(w)=1$, $f(w_s)=n-k-1$, and $f(w_i)=n-k$. Hence for any $w$ we have
$\std(w)=1-\frac{1 \times (n-k-1)}{n-k}=\frac{1}{n-k}>\rho$.
For instance, for $w=ba^{n-2}$, we have $\std(w)=1/2$. There are $2n-6=\Omega(n)$ such factors and hence at least these many $\rho$-overabundant words in $x$.
\end{proof}

\begin{corollary}\label{coro:common}
The $(\rho_1,\rho_2)$-common words in a word of length $n$ over an alphabet of size $\sigma$ are $\mathcal{O}((\sigma n)^2)$.
\end{corollary}
\begin{proof}
By Fact~\ref{fac:common} we know that $\std(w)$ is defined for $\mathcal{O}((\sigma n)^2)$ words. The $\rho_1$-avoided ones are $\cO(\sigma n)$~\cite{AMOB2017}, while the $\rho_2$-overabundant are $\cO(n)$ by Corollary~\ref{lem:3n}. Hence the  $(\rho_1,\rho_2)$-common words are $\mathcal{O}((\sigma n)^2)$.
\end{proof}

\section{Algorithm}\label{Section-Algorithm}


Based on Fact~\ref{abs} and Lemma~\ref{expl} all $\rho$-overabundant words of a word $x$ are factors of $x$ of the form $a y b$, where $a, b \in \Sigma$ and $y$ is the label of an explicit node of $\mathcal{T}(x)$. It thus suffices to consider these words and check for each of them whether it is $\rho$-overabundant. We can find the ones that have their longest proper prefix represented by an explicit node in $\mathcal{T}(x)$ easily, by taking the suffix-link from that node during a traversal of the tree. To find the ones that have their longest proper prefix represented by an implicit node we use the following fact, which follows directly from the definition of the suffix-links of the suffix tree.

\begin{fact}\label{subtree}
Suppose $aw$, where $a \in \Sigma$ and $w \in \Sigma ^*$, is a factor of a word $x$. Further suppose that $w$ is represented by an explicit node $v$ in $\mathcal{T} (x)$, while $aw$ by an implicit node along the edge $(u_1,u_2)$ in $\mathcal{T}(x)$. Then, 
the suffix-link from $u_2$ points to a node in the subtree of $\mathcal{T}(x)$ rooted at $v$.
\end{fact}

\alglanguage{pseudocode}
\begin{algorithm}[H]
\small
\caption{Compute all $\rho$-overabundant words}
\label{Algorithm:overabundant}
\begin{algorithmic}[1]
\Procedure{$\textsc{ComputeOverabundantWords}$}{word $x$, real number $\rho$}
    \State $\mathcal{T}(x) \leftarrow$ \textsc{BuildSuffixTree}($x$)
    \For{\mbox{each node  $v \in \mathcal{T}(x)$} }
        \State{$\mathcal{D}(v)\leftarrow \mbox{word-depth of }v$}
		\State{$\mathcal{C}(v)\leftarrow \mbox{number of terminal nodes in the subtree rooted at }v$}
	\EndFor
    \For {each  node $v \in \mathcal{T}(x)$} \Comment \textit{prefix node}
            \LineComment{Report $\rho$-overabundant words $w$ such that $w_p$ is explicit} 
 	 \State $u  \leftarrow \textit{suffix-link}[v]$  \Comment \textit{infix node}
     \If{$\mathcal{D}(v)>1$ \textbf{and} \textsc{IsInternal}($v$)}
            \State $f_p \leftarrow \mathcal{C}(v)$, $f_i \leftarrow  \mathcal{C}(u)$
              \If {$f_i > f_p$ \textbf{and} $u \neq \textsc{Root}(\mathcal{T}(x)$)} 
             	\For {each child $y$ of node $v$} 
                	\If {\textbf{not}(\textsc{IsTerminal}$(y)$ \text{~\bf{and}~} $\mathcal{D}(y) = \mathcal{D}(v)+1$)}
                	  	\State $f_w \leftarrow \mathcal{C}(y)$
   						\State $\alpha \leftarrow \mathcal{L}(y)[\mathcal{D}(v)+1]$
						\State $f_s \leftarrow \mathcal{C}(\Call{Child}{u,\alpha})$
						\State $E \leftarrow f_p\times f_s/f_i$
						\If{ $(f_w-E)/(\max\{1,\sqrt{E}\}) \geq \rho$} 
									\State \textsc{Report}($\mathcal{L}(y)[0\dd \mathcal{D}(v)]$)
						\EndIf
                	  \EndIf
                \EndFor
      		\EndIf
           \EndIf
                \LineComment{Report $\rho$-overabundant words $w$ such that $w_p$ is implicit} 
                \For {each child $y$ of node $v$} 
                	  \If {$\mathcal{D}(y) > \mathcal{D}(v)+1$ }
                	  	\If {\textsc{IsInternal}($y$)}
                	  		\State $z  \leftarrow \textit{suffix-link}[y]$
                	  	\Else \Comment \textit{$y$ is a  terminal node}
                	  		 \State $i  \leftarrow \textit{label}[y]$
                	  		 \State $z  \leftarrow \textit{node}[i+1]$  
                             \If {$\mathcal{D}(z) = \mathcal{D}(\textsc{Parent}(z))+1$ }
                             	\State {$z \leftarrow \textsc{Parent}(z)$}
                             \EndIf
                	  \EndIf
                	  \State $f_w \leftarrow  f_p \leftarrow \mathcal{C}(y)$
                	  \While{$\textsc{Parent}(z) \neq u$}
                	  	
                       	 	\State $f_i \leftarrow \mathcal{C}(\textsc{Parent}(z))$
							\State $f_s \leftarrow \mathcal{C}(z)$
							\State $E \leftarrow f_p\times f_s/f_i$
							\If{ $(f_w-E)/(\max\{1,\sqrt{E}\}) \geq \rho$} 
								\State \textsc{Report}($\mathcal{L}(y)[0 \dd \mathcal{D}[\textsc{Parent}(z)]+1]$)
							\EndIf
						\State {$z \leftarrow \textsc{Parent}(z)$}
					\EndWhile
                	\EndIf
                \EndFor
    \EndFor
\EndProcedure
\Statex
\end{algorithmic}
  \vspace{-0.4cm}%
  \label{algo:1}
\end{algorithm}

\clearpage

The algorithm first builds the suffix tree of word $x$, which can be done in time and space $\mathcal{O}(n)$ for words over an integer alphabet~\cite{farach1997optimal}. 
It is also easy to compute $\mathcal{D}(v)$ and $\mathcal{C}(v)$, for each node $v$ of $\mathcal{T} (x)$, within the same time complexity (lines $2-5$ in Algorithm~\ref{algo:1}).

The algorithm then performs a traversal of $\mathcal{T}(x)$.
When it first reaches a node $v$, it considers $\mathcal{L}(v)$ as a potential longest proper prefix of $\rho$-overabundant words---i.e.~$\mathcal{L}(v)=w_p=a w_i$, where $a \in \Sigma$. 
By following the suffix-link to node $u$, which represents the respective $w_i$, and based on the first letter of the label of each outgoing edge $(v,q)$ from $v$, it computes the deviation for all possible factors of $x$ of the form $w_p b$, where $b \in \Sigma$. (Note that we can answer all the $\textsc{Child}(u,\alpha)$ queries off-line in time $\cO(n)$ in total for integer alphabets.)
It is clear that this procedure can be implemented in time $\mathcal{O}(n)$ in total (lines $7-19$).

Then, while on node $v$ and based on Fact~\ref{subtree}, the algorithm considers for every outgoing edge $(v,q)$, the implicit nodes along this edge that correspond to words (potential $w_p$'s) whose proper longest suffix (the respective $w_i$) is represented by an explicit node in $\mathcal{T}(x)$. 

Hence, when $\mathcal{D}(q)-\mathcal{D}(v)>1$ the algorithm follows the suffix-link from node $q$ to node $z$. It then checks whether $\textsc{Parent}(z)=u$. If not, then the word $\mathcal{L}(q)[0 \dd \mathcal{D}(\textsc{Parent}(z))]$ is represented by an implicit node along the edge $(v,q)$ and hence $\mathcal{L}(q)[0 \dd \mathcal{D}(\textsc{Parent}(z))+1]$ has to be checked as a potential $\rho$-overabundant word. After the check is completed, the algorithm sets $z=\textsc{Parent}(z)$ and iterates until $\textsc{Parent}(z)=u$. This is illustrated in Figure~\ref{fig:algo}. By Theorem~\ref{revbound}, the $\textsc{Parent}(z)=u$ check will fail $\mathcal{O}(n)$ times in total. All other operations take time $\mathcal{O}(1)$ and hence this procedure takes time $\mathcal{O}(n)$ in total (lines $20-37$).

\begin{figure}[!t]
\begin{center}
\begin{adjustbox}{max width=0.99\textwidth}

\begin{tikzpicture}
\node[shape=circle,draw=black] (0) at (5,0) {};
\node[shape=circle,draw=white] (1) at (2,-1) {};
\node[shape=circle,draw=black, label = {left: node $v$}] (2) at (3,-3) {}; 
\node[shape=circle,draw=black, label = {right:  $u=\textit{suffix-link}[v]$}] (3) at (9,-3) {};  
\node[shape=circle,draw=white] (4) at (10,-1) {};

\draw [->,dashed,draw=red] (2) -- (3);

\node[shape=circle, ,minimum size=0.01cm,draw=black,fill=black] (A) at (3.5,-4) {};
\node[shape=circle, ,minimum size=0.1cm,draw=black,fill=black] (B) at (4,-5) {};
\node[shape=circle, ,minimum size=0.1cm,draw=black,fill=black] (C) at (4.5,-6) {};

\draw [->,decorate,decoration=snake] (0) -- (1);
\draw [->,decorate,decoration=snake] (0) -- (2);
\draw [->,decorate,decoration=snake] (0) -- (3);
\draw [->,decorate,decoration=snake] (0) -- (4);

\draw [decorate] (2) -- (A);
\draw [decorate] (A) -- (B);
\draw [decorate] (B) -- (C);
\node[shape=circle,draw=black] (5) at (1,-7) {};
\node[shape=circle,draw=black] (6) at (2,-7) {};
\node[shape=circle,draw=black,label= { left: \small{$q=\textsc{Child}(v, \alpha)$}}] (7) at (5,-7) {};
\draw [->,decorate] (2) -- (5);
\draw [->,decorate] (2) -- (6);

\draw [->,decorate] (C) -- (7);
\node[shape=circle,draw=white] (51) at (0.5,-8) {};
\node[shape=circle,draw=white] (52) at (1,-8) {};
\node[shape=circle,draw=white] (53) at (1.5,-8) {};
\draw [->,decorate] (5) -- (51);
\draw [->,decorate] (5) -- (52);
\draw [->,decorate] (5) -- (53);
\node[shape=circle,draw=white] (61) at (1.5,-8) {};
\node[shape=circle,draw=white] (62) at (2,-8) {};
\node[shape=circle,draw=white] (63) at (2.5,-8) {};
\draw [->,decorate] (6) -- (61);
\draw [->,decorate] (6) -- (62);
\draw [->,decorate] (6) -- (63);
\node[shape=circle,draw=white] (71) at (4.5,-8) {};
\node[shape=circle,draw=white] (72) at (5,-8) {};
\node[shape=circle,draw=white] (73) at (5.5,-8) {};
\draw [->,decorate] (7) -- (71);
\draw [->,decorate] (7) -- (72);
\draw [->,decorate] (7) -- (73);


\node[shape=circle,draw=white] (8) at (8.5,-3.5) {};
\node[shape=circle,draw=black] (8a) at (8,-4) {};
\node[shape=circle,draw=white] (8b) at (9.5,-3.5) {};
\node[shape=circle,draw=white] (8c) at (10,-3.5) {};
\node[shape=circle,draw=white] (8d) at (10.5,-3.5) {};
\draw [->,decorate] (3) -- (8a);
\draw [->,decorate] (3) -- (8b);
\draw [->,decorate] (3) -- (8c);
\draw [->,decorate] (3) -- (8d);
\node[shape=circle,draw=white] (9) at (8.5,-4.5) {};
\node[shape=circle,draw=black] (9a) at (8,-5) {};
\node[shape=circle,draw=white] (9b) at (9.5,-4.5) {};
\draw [->,decorate] (8a) -- (9);
\draw [->,decorate] (8a) -- (9a);
\draw [->,decorate] (8a) -- (9b);

\node[shape=circle,draw=white] (10) at (8.5,-5.5) {};
\node[shape=circle,draw=black] (10a) at (8,-6) {};
\node[shape=circle,draw=white] (10b) at (9.5,-5.5) {};
\draw [->,decorate] (9a) -- (10);
\draw [->,decorate] (9a) -- (10a);
\draw [->,decorate] (9a) -- (10b);

\node[shape=circle,draw=white] (11) at (8.5,-6.5) {};
\node[shape=circle,draw=black, label= { right: \small{$z=\textit{suffix-link}[q]$}}] (11a) at (8,-7) {};
\node[shape=circle,draw=white] (11b) at (9.5,-6.5) {};
\draw [->,decorate] (10a) -- (11);
\draw [->,decorate] (10a) -- (11a);
\draw [->,decorate] (10a) -- (11b);

\node[shape=circle,draw=white] (12) at (7.5,-8) {};
\node[shape=circle,draw=white] (12a) at (8,-8) {};
\node[shape=circle,draw=white] (12b) at (8.5,-8) {};
\draw [->,decorate] (11a) -- (12);
\draw [->,decorate] (11a) -- (12a);
\draw [->,decorate] (11a) -- (12b);

\draw [->,dashed,draw=red,thick] (2) -- (3);
\draw [->,dashed,draw=red,thick] (7) -- (11a);
\draw [->,dotted,draw=red,decoration=zigzag] (A) -- (8a);
\draw [->,dotted,draw=red,decoration=zigzag] (B) -- (9a);
\draw [->,dotted,draw=red,decoration=zigzag] (C) -- (10a);

\path [->,green,bend left]   (11a) edge (10a);
\path [->,green,bend left]   (10a) edge (9a);
\path [->,green,bend left]   (9a) edge (8a);
\path [->,green,bend left]   (8a) edge (3);
\end{tikzpicture}
\begin{tikzpicture}
\node[shape=circle,draw=black] (0) at (5,0) {};
\node[shape=circle,draw=white] (1) at (2,-1) {};
\node[shape=circle,draw=black, label = {left: node $v$}] (2) at (3,-3) {}; 
\node[shape=circle,draw=black, label = {right:  $u=\textit{suffix-link}[v]$}] (3) at (9,-3) {};  
\node[shape=circle,draw=white] (4) at (10,-1) {};

\draw [->,dashed,draw=red] (2) -- (3);

\node[shape=circle, ,minimum size=0.01cm,draw=black,fill=black] (A) at (3.5,-4) {};
\node[shape=circle, ,minimum size=0.1cm,draw=black,fill=black] (B) at (4,-5) {};
\node[shape=circle, ,minimum size=0.1cm,draw=black,fill=black] (C) at (4.5,-6) {};

\draw [->,decorate,decoration=snake] (0) -- (1);
\draw [->,decorate,decoration=snake] (0) -- (2);
\draw [->,decorate,decoration=snake] (0) -- (3);
\draw [->,decorate,decoration=snake] (0) -- (4);

\draw [decorate] (2) -- (A);
\draw [decorate] (A) -- (B);
\draw [decorate] (B) -- (C);
\node[shape=circle,draw=white] (5) at (1,-4) {};
\node[shape=circle,draw=white] (6) at (2.5,-4) {};
\node[shape=circle,draw=black,double, label= { left: \small{$q=\textsc{Child}(v, \alpha)$, $\textit{label}[q]=i$}}] (7) at (5,-7) {};
\draw [->,decorate] (2) -- (5);
\draw [->,decorate] (2) -- (6);
\draw [->,decorate] (C) -- (7);


\node[shape=circle,draw=black] (8a) at (8,-4) {};
\node[shape=circle,draw=white] (8) at (8.5,-3.5) {};
\node[shape=circle,draw=white] (8b) at (9,-3.5) {};
\node[shape=circle,draw=white] (8c) at (9.5,-3.5) {};
\node[shape=circle,draw=white] (86) at (10,-3.5) {};
\draw [->,decorate] (3) -- (8a);
\draw [->,decorate] (3) -- (8b);
\draw [->,decorate] (3) -- (8c);
\draw [->,decorate] (3) -- (8d);
\node[shape=circle,draw=white] (9) at (8.5,-4.5) {};
\node[shape=circle,draw=black] (9a) at (8,-5) {};
\node[shape=circle,draw=white] (9b) at (9.5,-4.5) {};
\draw [->,decorate] (8a) -- (9);
\draw [->,decorate] (8a) -- (9a);
\draw [->,decorate] (8a) -- (9b);

\node[shape=circle,draw=white] (10) at (8.5,-5.5) {};
\node[shape=circle,draw=black] (10a) at (8,-6) {};
\node[shape=circle,draw=white] (10b) at (9.5,-5.5) {};
\draw [->,decorate] (9a) -- (10);
\draw [->,decorate] (9a) -- (10a);
\draw [->,decorate] (9a) -- (10b);

\node[shape=circle,draw=white] (11) at (8.5,-6.5) {};
\node[shape=circle,draw=black, double, label= { right: \small{$\textit{label}[z]=i+1$}}] (11a) at (8,-7) {};
\node[shape=circle,draw=white] (11b) at (9.5,-6.5) {};
\draw [->,decorate] (10a) -- (11);
\draw [->,decorate] (10a) -- (11a);
\draw [->,decorate] (10a) -- (11b);

\draw [->,dashed,draw=red,thick] (2) -- (3);
\draw [->,dotted,draw=red,decoration=zigzag] (A) -- (8a);
\draw [->,dotted,draw=red,decoration=zigzag] (B) -- (9a);
\draw [->,dotted,draw=red,decoration=zigzag] (C) -- (10a);

\path [->,green,bend left]   (11a) edge (10a);
\path [->,green,bend left]   (10a) edge (9a);
\path [->,green,bend left]   (9a) edge (8a);
\path [->,green,bend left]   (8a) edge (3);
\end{tikzpicture}

\end{adjustbox}
\end{center}
\caption{The above figures illustrate the nodes (implicit or explicit) considered in a step (lines 6-37) of Algorithm~\ref{algo:1}. The figure on the left presents the case where $\textsc{Child}(v, \alpha)$ is an internal node, while the right one the case that it is a leaf. Black nodes represent implicit nodes along the edge $(v,q)$ that we have to consider as potential $w_p$, and the red dotted line joins them with the respective (white) explicit node that represents the longest suffix of this $w_p$, i.e.~$w_i$.}
\label{fig:algo}
\end{figure}
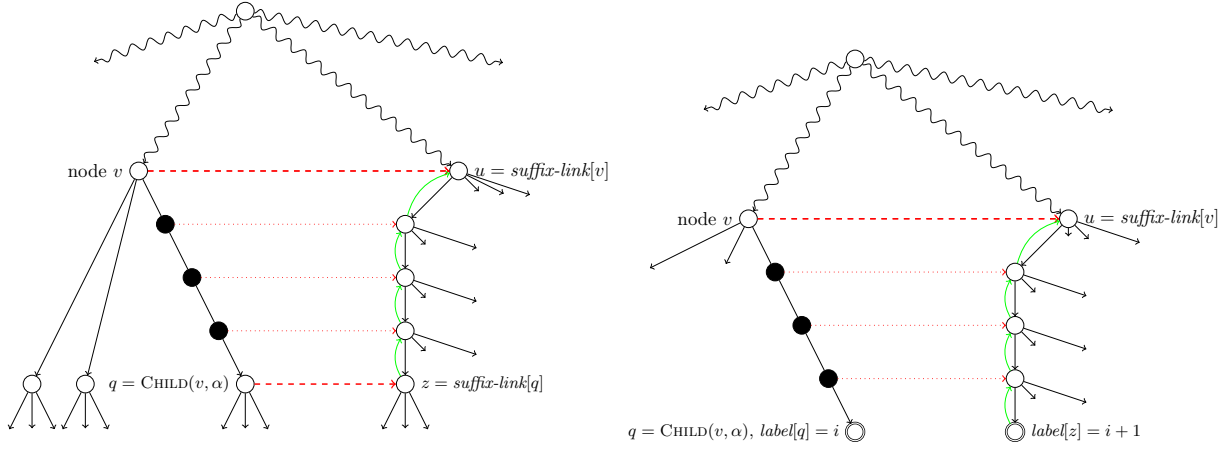

We formalise this procedure in Algorithm~\ref{algo:1}, where we assume that the suffix tree of $x\$$ is built, where $\$$ is a special letter, $\$ \notin \Sigma$. This forces all terminal nodes in $\mathcal{T}(x)$ to be leaf nodes. We thus obtain the following result; optimality follows directly from Lemma~\ref{lem:opt}.

\begin{theorem}
Algorithm~\ref{algo:1} solves problem \textsc{AllOverabundantWordsComputation} in time and space $\mathcal{O}(n)$, and this is time-optimal.
\end{theorem}

\section{Experimental Results: Effectiveness, Efficiency, and Applications}\label{sec:exp}

Algorithm~\ref{algo:1} was implemented as a program to compute the $\rho$-overabundant words in one or more input sequences. The program was implemented in the \texttt{C++} programming language. 
Our program makes use of the implementation of the {\em compressed suffix tree} available in the Succinct Data Structure Library~\cite{gbmp2014sea}. The input parameters are a (Multi)FASTA file with the input sequence(s) and a real number $\rho > 0$. The output is a file with the set of $\rho$-overabundant words per input sequence. The implementation is distributed under the GNU General Public License, and it is available at \url{http://github.com/solonas13/aw}. The experiments were conducted on a Desktop PC using one core of Intel Core i5-4690 CPU at 3.50GHz. 
The program was compiled with \texttt{g++} version 4.8.4 at optimisation level 3 (-O3). We also implemented a brute-force approach to confirm the correctness of our implementation. 

{\it Experiment I. (Effectiveness)} In the first experiment, our task was to establish the effectiveness of the statistical model in identifying overabundant words. To this end, we generated 25 random sequences of length $n=80,000$ over the DNA alphabet $\Sigma=\{\texttt{A}, \texttt{C}, \texttt{G}, \texttt{T}\}$ (uniform distribution). Then for each of these sequences, we inserted a random word $w$ of length $m=6$ in $t$ random positions. We varied the value of $t$ based on the fact that in a random sequence of length $n$ over an alphabet of size $\sigma=|\Sigma|$, where letters are independent, identically uniformly distributed random variables, a specific word of length $m$ is expected to occur roughly $r=n/\sigma^m$ times. We hence considered $t$ equal to $r$, $2r$, $4r$, $8r$, and $16r$. We then ran our program for each resulting sequence to identify the $\rho$-overabundant words with $\rho=0.000001$, and output the deviation of the inserted word $w$, as well as the word $w_{\max}$ with the maximum deviation. The inserted word $w$ was reported as a $\rho$-overabundant word in \emph{all} cases. Furthermore, in many cases the word with the maximum deviation was $w$ itself and in many other cases one of its factors; this was true in {\em all} cases for $t \geq 80 \approx 4r$. Hence, the model is effective in identifying words that are overabundant. The full results of this experiment are presented in Table~\ref{tab:effectiveness}. 

\begin{table}[!t]
	\begin{center}
        \begin{tabular}{|c|c|c|c|c|c|}
\hline
Times $t$ of inserting $w$ & 20 & 40 & 80 & 160 & 320 \\ \hline
$w$ & \texttt{TTACAA} &  \cellcolor{green!20} \texttt{GTGCCC} & \cellcolor{green!20} \texttt{CACTTT} & \cellcolor{green!20} \texttt{AGTTAC} & \cellcolor{green!20} \texttt{AAACAG} \\ 

$\std(w)$ & \texttt{2.233313} & \cellcolor{green!20} \texttt{4.143015} & \cellcolor{green!20} \texttt{5.623615} & \cellcolor{green!20} \texttt{6.010327} & \cellcolor{green!20} \texttt{5.674220} \\ \hdashline
                            
$w_{\max}$ & \texttt{CTCCTATG} & \cellcolor{green!20} \texttt{GTGCCC} & \cellcolor{green!20} \texttt{CACTTT} & \cellcolor{green!20} \texttt{AGTTA} & \cellcolor{green!20} \texttt{ACAG} \\ 
                             
$\std(w_{\max})$ &  \texttt{3.354102} & \cellcolor{green!20} \texttt{4.143015} & \cellcolor{green!20} \texttt{5.623615} & \cellcolor{green!20} \texttt{6.900740} & \cellcolor{green!20}\texttt{9.617803} \\ \hline

$w$ & \texttt{AATCTG} & \texttt{AGTCGA} & \cellcolor{green!20}\texttt{GAAGTC} & \cellcolor{green!20} \texttt{TATCTT} & \cellcolor{green!20} \texttt{CAAAAA} \\  

$\std(w)$ & \texttt{2.034233} & \texttt{2.888529} & \cellcolor{green!20} \texttt{4.456468} & \cellcolor{green!20} \texttt{5.073860} & \cellcolor{green!20} \texttt{11.071170} \\ \hdashline
                            
$w_{\max}$ & \texttt{ATTGGGG} & \texttt{TCTGTATG} & \cellcolor{green!20} \texttt{GAAGTC} & \cellcolor{green!20} \texttt{ATCTT} & \cellcolor{green!20} \texttt{CAAAAA} \\  
                              
$\std(w_{\max})$ &  \texttt{3.265609} & \texttt{3.272727} & \cellcolor{green!20} \texttt{4.456468} & \cellcolor{green!20} \texttt{6.115612} & \cellcolor{green!20} \texttt{11.071170} \\ \hline

$w$ & \texttt{GTACCA} & \texttt{GGCGTG} & \cellcolor{green!20} \texttt{AAGGAT} & \cellcolor{green!20} \texttt{GGGTCC} & \cellcolor{green!20} \texttt{TTCCGG} \\ 

$\std(w)$ & \texttt{2.187170} & \texttt{3.658060} & \cellcolor{green!20} \texttt{4.428189} & \cellcolor{green!20} \texttt{5.467296} & \cellcolor{green!20} \texttt{5.256409} \\ \hdashline
                            
$w_{\max}$ & \texttt{TCTGTGCG} & \texttt{ACGATACC} & \cellcolor{green!20} \texttt{AAGGAT} & \cellcolor{green!20} \texttt{GGTCC} & \cellcolor{green!20} \texttt{TTCCG} \\ 
                              
$\std(w_{\max})$ &  \texttt{3.548977} & \texttt{4.000000} & \cellcolor{green!20} \texttt{4.428189} & \cellcolor{green!20} \texttt{6.787771} & \cellcolor{green!20} \texttt{9.105009} \\ \hline

$w$ & \texttt{CCATAG} & \texttt{GTTGAT} & \cellcolor{green!20} \texttt{TGAGCG} & \cellcolor{green!20} \texttt{ACATTT} & \cellcolor{green!20} \texttt{CTTGTA} \\ 

$\std(w)$ & \texttt{2.470681} & \texttt{2.467858} & \cellcolor{green!20} \texttt{4.214544} & \cellcolor{green!20} \texttt{5.755475} & \cellcolor{green!20} \texttt{5.362435} \\ \hdashline
                            
$w_{\max}$ & \texttt{CAGTGGTC} & \texttt{TTTTCCT} & \cellcolor{green!20} \texttt{TGAGC} & \cellcolor{green!20} \texttt{ACATT} & \cellcolor{green!20} \texttt{TTGTA} \\ 
                              
$\std(w_{\max})$ &  \texttt{3.333333} & \texttt{3.368226} & \cellcolor{green!20} \texttt{5.072968} & \cellcolor{green!20} \texttt{6.376277} & \cellcolor{green!20} \texttt{9.467110} \\ \hline

$w$ & \texttt{TCGACA} & \texttt{CGCTTT} & \cellcolor{green!20} \texttt{TACAAC} & \cellcolor{green!20} \texttt{TATTAG} & \cellcolor{green!20} \texttt{TGAGAT} \\  

$\std(w)$ & \texttt{1.531083} & \texttt{2.789220} & \cellcolor{green!20} \texttt{3.552902} & \cellcolor{green!20} \texttt{4.959926} & \cellcolor{green!20} \texttt{5.124976} \\ \hdashline
                            
$w_{\max}$ & \texttt{CTTTGCT} & \texttt{ATTACC} & \cellcolor{green!20} \texttt{ACAAC} & \cellcolor{green!20} \texttt{ATTAG} & \cellcolor{green!20} \texttt{GACAT} \\ 
                             
$\std(w_{\max})$ &  \texttt{3.308195} & \texttt{3.322163} & \cellcolor{green!20} \texttt{5.653479} & \cellcolor{green!20} \texttt{6.837628} & \cellcolor{green!20} \texttt{10.012316} \\ \hline
\end{tabular}
\newline
\linebreak
    \caption{The deviation of the randomly generated inserted word $w$, as well as the word $w_{\max}$ with the maximum deviation. The length of each of the 25 randomly generated sequences over $\Sigma=\{\texttt{A}, \texttt{C}, \texttt{G}, \texttt{T} \}$ was $n=80,000$, the length of $w$ was $m=6$, and $\rho=0.000001$. In green are the cases when the word with the maximum deviation was $w$ itself or one of its factors.}
	\label{tab:effectiveness}
	\end{center}
\end{table} 

{\it Experiment II. (Efficiency)} Our task here was to establish the fact that the elapsed time of the implementation grows linearly with $n$, the length of the input sequence.  As input datasets, for this experiment, we used synthetic DNA ($\sigma=4$) and proteins ($\sigma=20$) sequences ranging from $1$ to $128$ M (Million letters). For each sequence we used a constant value of $\rho=10$. The results are plotted in Fig.~\ref{fig:linear}. It becomes evident from the results that the elapsed time of the program grows linearly with $n$. The longer time required for the proteins sequences compared to the DNA sequences for increasing $n$ is explained by the dependence of the time required to answer queries of the from $\textsc{Child}(v,\alpha)$ on the size of the alphabet ($\sigma=20$ {\em vs}.~$\sigma=4$) in the implementation of the compressed suffix tree we used.
\begin{figure}[!t]
	\begin{center}
    \includegraphics[width=0.44\textwidth,angle=270]{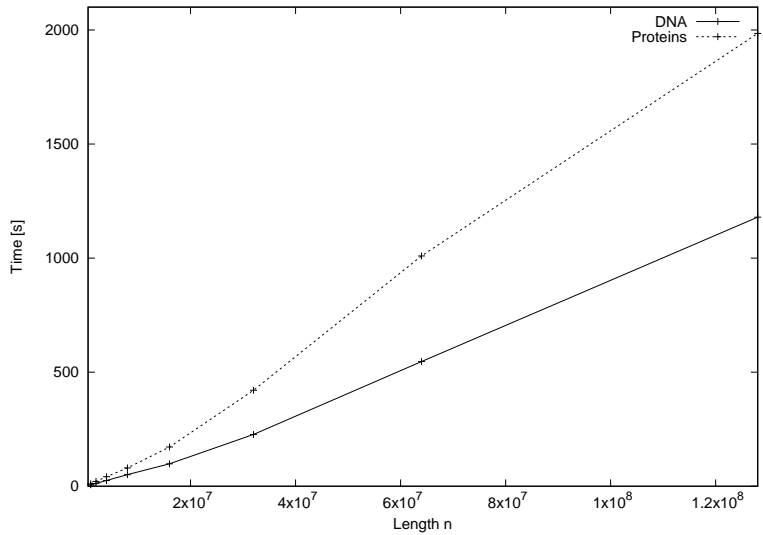}
	\caption{Elapsed time of Algorithm~\ref{algo:1} using synthetic DNA ($\sigma=4$) and proteins ($\sigma=20$) sequences of length $1$M to $128$M.}
	\label{fig:linear}
	\end{center}
\end{figure}

{\it Experiment III. (Real Application)} Here we proceed to the examination of seven collections of Conserved Non-coding Elements (CNEs) obtained through multiple sequence alignment between the human and other genomes. Despite being located at the non-coding part of genomes, CNEs can be extremely conserved on the sequence level across organisms. Their genesis, functions and evolutionary dynamics still remain enigmatic~\cite{polychronopoulos2014conserved,harmston2013mystery}. The detailed description of how those CNEs were identified can be found in~\cite{Polychronopoulos201479}. For each CNE of these datasets, a sequence stretch (surrogate sequence) of non-coding DNA of equal length and equal \texttt{GC} content was taken at random from the repeat-masked human genome. The CNEs of each collection were concatenated into a single long sequence and the same procedure was followed for the corresponding surrogates. We have determined through the proposed algorithm the overabundant words for $k = 10$ (decamers) and $\rho = 3$ for these fourteen datasets and the results are presented in Table~\ref{tab:Exp3-a}. Likewise, in Table~\ref{tab:Exp3-b}, we show all overabundant words (i.e.~$k > 2$) for $\rho = 3$.

\begin{table}[!t]
	\begin{center}
        \begin{tabular}{|c|c|c|c|c|c||c|c|}
\hline
$k=10,$	&~CNEs~ 	&~CNEs~ 	&~CNEs~ &	~CNEs~ &	~CNEs~ &	\small{Mammalian} &	\small{Amniotic}\\ 
$\rho=3$	&75-80	&80-85	&85-90&	90-95&	95-100&	&	\\ \hline
Surr	&1,144	&718	&473	&297	&469	&15,470	&2,874\\ \hline
CNEs&	331	&181&	100&	59&	71&	491&	149\\ \hline
 \bf{Ratio}	&\bf{3.46}&	\bf{3.97}&	\bf{4.73}&	\bf{5.03}&	\bf{6.61}&	\bf{31.51}&	\bf{19.29}\\ \hline
\end{tabular}
\newline
\linebreak
    \caption{Number of overabundant words for $k = 10$ and  $\rho = 3$.}
	\label{tab:Exp3-a}
\end{center}
\end{table}

\begin{table}[!t]
	\begin{center}
        \begin{tabular}{|c|c|c|c|c|c||c|c|}
\hline
$k >2,$	&~CNEs~ 	&~CNEs~ 	&~CNEs~ &	~CNEs~ &	~CNEs~ &	\small{Mammalian} &	\small{Amniotic}\\ 
$\rho=3$	&75-80	&80-85	&85-90&	90-95&	95-100&	&	\\ \hline
Surr	&5,925	&3,798	&2,770	&1,948	&2,405	&69,022	&12,913\\ \hline
CNEs&	1,373	&778	&512	&390	&403	&7,549&	1,401\\ \hline
 \bf{Ratio}	&\bf{4.32} &	\bf{4.88}&	\bf{5.41}&	\bf{4.99} &\bf{5.97}&	\bf{9.14}&\bf{9.22}\\ \hline
\end{tabular}
\newline
\linebreak
    \caption{Number of overabundant words for $k > 2$ and  $\rho = 3$.}
	\label{tab:Exp3-b}
\end{center}
\end{table}

The first five CNE collections have been composed through multiple sequence alignment of the same set of genomes (human vs. chicken; mapped on the human genome) and they differ only in the thresholds of sequence similarity applied between the considered genomes: from 75\% to 80\% (the least conserved CNEs, which thus are expected to serve less demanding functional roles) to 95–100\% which represent the extremely conserved non-coding elements (UCNEs or CNEs 95–100)~\cite{Polychronopoulos201479}. The remaining two collections have been composed under different constraints and have been derived after alignment of Mammalian and Amniotic genomes. In Tables~\ref{tab:Exp3-a} and~\ref{tab:Exp3-b}, the last line shows the ratios formed by the numbers of overabundant words of each concatenate of surrogates divided by the numbers of overabundant words of the corresponding CNE dataset. Two immediate results stem from inspection of Tables~\ref{tab:Exp3-a} and~\ref{tab:Exp3-b}:
\begin{enumerate}
\item In all cases, the number of overabundant words from the surrogate concatenate of sequences {\em far exceeds} the corresponding number derived from the CNE dataset.
\item In the case of datasets with increasing degree of similarity between aligned genomes (from 75-80 to 95-100), the ratios of the numbers of overabundant words show a clear, {\em increasing trend}.
\end{enumerate}

Both these findings can be understood on the basis of the difference in functionality between CNE and surrogate datasets. As we briefly describe in Section~\ref{sec:intro}, this systematic difference (finding 1 above) is expected on the basis of the self-enhancing elongation of relatively long homonucleotide tracts~\cite{Hile2004745,Levinson01051987}, which occurs mainly in the non-constrained parts of the genome, here the surrogate datasets. Moreover, finding 2 corroborates the proposed mechanism of overabundance, as in CNE datasets 1-5 depletion in overabundant words quantitatively follows the degree of sequence conservation. Inspection of the individual overabundant words found in the surrogate datasets verifies that they largely consist of short repeats of the types described in~\cite{Hile2004745} and in~\cite{Levinson01051987}. There is an analogy of this finding with a corresponding one, concerning the occurrence of avoided words in the same sequence sets, which was described in~\cite{AMOB2017}.

\section{Final Remarks}\label{sec:conc}
By Lemmas~\ref{lem:3n} and~\ref{lem:opt}, we know that the maximum number $\textsf{OW}(n,\sigma)$ of overabundant words in any sequence of length $n$ over an alphabet of size $\sigma>1$ lies between $2n-6 \leq \textsf{OW}(n,\sigma) \leq 3n-2-2\sigma$. We have conducted computational experiments, and for $\sigma > 2$ we obtained sequences with more than $2n$ overabundant words. An open problem is to find $\textsf{OW}(n,\sigma)$.

\bibliographystyle{plain}
\bibliography{references}

\end{document}